\documentclass[12points,a4paper]{amsart}
\usepackage{graphicx}
\usepackage[utf8]{inputenc} 
\usepackage{xcolor}
\usepackage{amsmath}
\usepackage{amsthm}
\usepackage{amssymb}
\usepackage{fullpage}
\usepackage{multirow}
\usepackage{enumerate}
\usepackage{hyperref}

\def\lin{\operatorname{lin}}
\def\inter{\operatorname{int}}

\newtheorem{Th}{Theorem}
\newtheorem{Prop}[Th]{Proposition}
\newtheorem{Cor}[Th]{Corollary}

\theoremstyle{remark}

\newtheorem{R}{Remark}

\theoremstyle{definition}
\newtheorem{Df}{Definition}
\newtheorem{Ex}{Example}

\title{Beyond morphophoricity: $s$-tight IC measurements in geometric generalised probabilistic theories.}
\author{Anna Szymusiak}
\address{Faculty of Mathematics and Computer Science, Jagiellonian University, \L ojasiewicza 6, 30-348 Krak\'{o}w, Poland}
\email{anna.szymusiak@uj.edu.pl}

\begin{document}

\begin{abstract}
    The analysed in this paper new class of $s$-tight IC measurements contains both morphophoric measurements, preserving the geometry of the states space, and tight IC measurements, introduced nearly 20 years ago by Scott in the quantum case as optimal for the task of linear quantum tomography. By looking at the mathematical side of these classes we discover their common feature, which is also  preserved in the broader class of $s$-tight IC measurements: a particularly elegant form of the formula that can be seen as the generalised form of the Urgleichung known from the QBist approach to quantum theory. In particular, the tight IC measurements are identified as the ones for which this generalised Urgleichung takes an exceptionally simple form.
\end{abstract}
\maketitle
\section{Introduction}
The QBist approach \cite{Appetal11,Appetal17,DeBetal21,Fucetal22,FucSch09,FucSch11,FucSch13,Ros11} to the basis of quantum theory canonically distinguishes  SIC-POVMs  as exceptional measurements and focuses on two key features:\begin{enumerate}[i)]
    \item the geometry of the quantum state space is the same as the geometry of the \emph{qplex}, i.e. the set of possible probability distributions of the measurement outcomes, 
    \item the probabilities of any measurement (`on the ground') can be computed in a simple and elegant way from the probabilities of the SIC measurement (`in the sky') -- the formula is known as the `Urgleichung' or `primal equation'.
\end{enumerate}
Previously, we have considered the following generalisations of the above statements. Firstly \cite{SloSzy20}, we identified a class of morphophoric POVMs, i.e. a class of quantum measurements that preserve the geometry of the quantum state space. It turned out that the morphophoric version of the primal equation is still quite simple and elegant. Next \cite{SzySlo23}, we extended our approach from quantum theory to the GGPTs, geometric generalised probabilistic theories, showing that the basic ideas of QBism are not limited to the quantum world. This time in looking for further generalisations we take as a starting point the requirement that the probabilities of any measurement `on the ground' can be computed in a \emph{painless} way from the probabilities of the measurement `in the sky'. There is no coincidence that the word `painless' appears also in the title of the paper introducing the tight frames for the very first time \cite{Dauetal86}. Similarly as in the case of morphophoric measurements, tight frames will play a central role in characterising measurements fulfilling our requirements.    

An observation that turned up to be crucial in the previous works and remains this way is that both key features focus solely on what happens on the state space which is just a subset of an affine hyperplane of the whole space. In consequence, the essential properties of the measurement should reveal themselves in the \emph{traceless} parts of its effects, i.e. the projections of the effects onto the traceless hyperplane, if we talk about quantum theory, or onto the corresponding subspace in case of GGPTs. While in the morphophoric case these projections formed a tight frame, this time we just require them to be \emph{scalable} to a tight frame. Such approach is not entirely new since Scott already analysed a special case of such scalability in quantum case under a name of \emph{tight IC} POVMs \cite{Sco06}. The overlap of the classes of tight IC and morphophoric POVMs is nonempty: e.g. the 2-design POVMs fall into both definitions. Thus it is  natural to think of a class of measurements that would contain both these classes and preserve some of their properties even in a limited form. The general scalability property satisfy both these requirements, in or outside the quantum world.

The paper is organised as follows. Section \ref{GGPT}  covers preliminary topics in geometric generalised probabilistic theories (GGPTs). This notion has been introduced in \cite{SzySlo23} for GPTs such that the state space is equipped with some inner product, as geometry requires the notion of orthogonality.  Section \ref{framesintro} provides a necessary background in the frame theory. In Section \ref{classes} we introduce new classes of measurements, analyse their basic properties and give some examples in quantum setup. Finally, in Section \ref{genurg} we derive some generalisations of the primal equation for the newly introduced classes of measurements. In particular, we identify tight IC measurements to be the only ones for which the generalised Urgleichung with the canonical instrument involved takes the simple form of the classical total probability law plus a correction term.

\section{Geometric GPTs}\label{GGPT}

 The main idea behind generalised probabilistic theories is to provide a theoretical framework in which a physical theory can be treated as a probabilistic theory. It means that such theory needs to provide us some rules by which we could predict the statistics of the outcomes of an experiment if we know precisely how the physical system was prepared. Thus the basic ingredients of any GPT are states, understood as possible preparation procedures, effects, describing measurement apparatus resulting in `yes' or `no' outcomes, and a function assigning probabilities to the pairs of states and effects, representing the probability of obtaining the `yes' outcome. Ludwig's embedding theorem \cite[IV, Theorem 3.7]{Lud85} guarantees that under some minimal constraints ensuring some kind of physical realism, GPTs can be mathematically described in terms of order vector spaces. Although in general such space can be of infinite dimension, it is often assumed that its dimension is finite which can be interpreted as assuming that only finite number of parameters is required to uniquely describe the state of the system. While this operational approach to physical theory can be traced back to 1970's \cite{DavLew70,Lud70,Dav76}, it has been experiencing a revival for the last 20 years, see e.g. \cite{Bar07,BarWil11,Udu12,JanLal13,Baretal14,BarWil16,MulMas16,Lam17,Wil18,Wil19,Pla21,Wet21,Tak22}. For the more detailed introduction to the geometric GPTs see \cite{SzySlo23}.

An \emph{abstract state space} is a triple $(V,C,e)$ consisting of a (finite-dimensional) real vector space $V$ ordered by a convex cone $C$ with a distinguished strictly positive linear functional $e\in V^*$ called the \emph{unit effect}. The cone $C$ is assumed to be proper ($C\cap(-C)=\{0\}$), generating ($C-C=V$) and closed. The strict positivity of $e$ means that $e(x)>0$ for every $x>0$, i.e. for every $x\in C\setminus\{0\}$. In other words, $e\in\inter C^*$, where $C^*:=\{g\in V^*:g(x)\geq 0 \textnormal{ for every } x\in C\}$ is the dual cone of positive functionals. The 1-level set of the unit effect intersected with the cone $C$ forms a base $B:=\{x\in C: e(x)=1\}$ of $C$. The base $B$ is interpreted as the set of states of the system and the extreme points of $B$ as pure states. The effects are described by positive functionals $g\in C^*$ such that $g\leq e$. We assume the \emph{no restriction hypothesis} which means that all such effects are considered to be physically possible. Now, if the system is in the state $x$, the probability of obtaining the answer `yes' to the question represented by the effect $g$ is simply given by $g(x)\in[0,1]$. In particular, the unit effect $e$ corresponds to the trivial measurement that always gives the `yes' outcome (can be think of as asking `is there a system?'). Finally, by a \emph{measurement} we mean a sequence of nonzero effects $(\pi_j)_{j=1}^n$ such that $\sum_{j=1}^n\pi_j=e$. In consequence, $\pi_j(x)$  represents the probability that the measurement outcome is $j$ given that the pre-measurement state is $x\in B$, and the \emph{measurement map} $\pi:B\ni x\mapsto (\pi_j(x))_{j=1}^n\in\mathbb R^n$ maps the set of states into the probability simplex $\Delta_n:=\{p\in\mathbb R^n:p_j\geq 0 \textnormal{ for }j=1,\ldots,n \textnormal{ and }\sum_{j=1}^np_j=1\}$.  

Throughout the paper we will assume that the GPT in question is additionally equipped with some 
geometry, which is introduced first as an inner product $\langle\cdot,\cdot\rangle_0$ in the vector subspace $V_0:=e^{-1}(0)$ and then extended to the inner product in the whole space $V$ by choosing $m\in \inter B$ indicating the orthogonal direction to $V_0$ and the \emph{size parameter} $\mu>0$ specifying the squared length of vector $m$. In consequence, we obtain the inner product $\langle\cdot,\cdot\rangle$ in $V$ given by
\begin{equation*}
   \langle x,y\rangle=\langle x-e(x)m,y-e(y)m\rangle_0+e(x)e(y)\mu, \quad x,y\in V. 
\end{equation*}
Thus by a \emph{geometric generalised probabilistic theory} (\emph{GGPT}) we mean a 6-tuple $(V,C,e,m,\mu,\langle\cdot,\cdot\rangle_0)$. 
The inner product $\langle\cdot,\cdot\rangle_0$ describes fully the geometry of the set of states since $V_0=V_1-V_1$, where $V_1:=e^{-1}(1)$ is the affine subspace generated by $B$. Therefore it can be seen as some intrinsic property of the theory, the \emph{internal} geometry, while $m$ and $\mu$ describe the \emph{external} geometry, positioning $B$ in the whole space $V$. The choice of $m$ can be arbitrary but whenever there is a state in $B$ that is in some way distinguished, e.g. by the symmetries, it may be natural to choose such state as $m$. Finally, the size parameter $\mu$ can be chosen so that the inner product and the order structure are compatible. There are three types of such compatibility, expressed by the relations between the cone $C$ and the \emph{positive dual} cone of $C$ in $V$ defined by $C^+:=\{y\in V:\langle x,y\rangle\geq 0 \textnormal{ for all }x\in C\}$: $C$ is called \emph{infra-dual} if $C\subset C^+$, \emph{supra-dual} if $C^+\subset C$ and \emph{self-dual} if $C^+=C$. In such cases we can also say that the GGPT in question is infra- supra- or self-dual, respectively. It turns out, that by choosing the parameter $\mu$ sufficiently small or sufficiently big we can always make $C$ supra- or infra-dual, respectively. However, it is not always possible to obtain self-duality. For the purposes of this paper we assume that the parameter $\mu$ is set to guarantee supra-duality or, whenever possible, self-duality of the GGPT.

The inner product in $V$ induces the isometric linear isomorphism $T:V\to V^*$ given by $T(y)(x):=\langle x,y\rangle$. In particular, $T(C^+)=C^*$ thus if the GGPT is self-dual $T$ is also an order isomorphism. Note that $T^{-1}(e)=\frac{1}{\mu}m$. By the supra-duality assumption, $T$ allows us to identify measurement effects with (unnormalised) states, namely, for a measurement $\pi=(\pi_j)_{j=1}^n$ we put $v_j:=T^{-1}(\pi_j)\in C$ for $j=1,\ldots,n$. After normalisation we obtain a family of states $w_j:=v_j/e(v_j)=\frac{\mu}{\pi_j(m)}T^{-1}(\pi_j)\in B$ for $j=1,\ldots,n$.

As already mentioned in the introduction, since we focus on the set of states $B$ which is a subset of an affine hyperplane $V_1$, both the corresponding linear supspace $V_0$ and its image under the isomorphism $T$ denoted by $V_0^*:=T(V_0)=\{f\in V^*:f(m)=0\}$ should play a crucial role in revealing key features of the measurements. Let us denote by $P_0$ and $\mathcal P_0$ the orthogonal projections onto $V_0$ and $V_0^*$, respectively, i.e. $P_0:V\ni x\mapsto x-e(x)m\in V_0$ and $\mathcal P_0:V^*\ni f\mapsto f-f(m) e\in V^*_0$. Obviously, $\mathcal P_0\circ T=T\circ P_0$. 

With a GGPT we can also associate another constant, related strictly to the geometry of the set of states, namely $\chi:=\max\{\|x-m\|_0^2:x\in B\}$. Indeed, the value of $\chi$ is not affected by the choice of $\mu$. It is also easy to see that the states that maximise the distance from $m$ are necessarily pure. In particular, the GGPTs such that all pure states maximise this distance are called \emph{equinorm}.

\section{Frames, tight frames and scalable frames}\label{framesintro}

Although frames can be defined for any separable Hilbert space, for our purposes it suffices to focus on finite dimensional case. Moreover, for the same reason, we can restrict our attention to finite frames. In consequence, we do not need to discuss technical issues that arise in infinite cases. In the following we briefly present some basic facts concerning finite frames that can be found e.g. in \cite{Wal18}. The concept of scalable frames has been introduced in \cite{Kutetal13}.

Let us denote by $\mathcal H$  a $d$-dimensional Hilbert space over the field $\mathbb K$. 
\begin{Df}
We say that a sequence $H:=(h_j)_{j=1}^n$ of vectors in $\mathcal H$ is a {\bf frame} for $\mathcal H$ if there exist $\alpha,\beta>0$ such that \begin{equation}\label{frame}
    \alpha\|h\|^2\leq\sum_{j=1}^n|\langle h,h_j\rangle|^2\leq \beta\|h\|^2
\end{equation}
for all $h\in\mathcal H$. We refer to $\alpha$ and $\beta$ as {\bf lower} and {\bf upper frame bounds}, respectively.
\end{Df}
It is easy to see that $H$ is in fact nothing more than a spanning sequence. However, by writing it down in the above form, we can immediately distinguish a special case, in which the lower and upper bounds coincide:

\begin{Df}
We say that $H:=(h_j)_{j=1}^n$ is a {\bf tight frame} for $\mathcal H$ if there exists $\alpha>0$ such that \begin{equation}\label{tightframe}
    \sum_{j=1}^n|\langle h,h_j\rangle|^2= \alpha\|h\|^2
\end{equation}
for all $h\in\mathcal H$. We refer to $\alpha$ as the {\bf  frame bound}.
\end{Df}

With a frame $H$ we can associate the following operators:
\begin{description}
    \item[analysis operator] $X:\mathcal H\to\mathbb K^n$ given by $(Xh)_j:=\langle h,h_j\rangle$,
    \item[synthesis operator] $X^*:\mathbb K^n\to\mathcal H$ given by $X^*\gamma=\sum_{j=1}^n\gamma_jh_j$,
    \item[frame operator] $S:=X^*X:\mathcal H\to\mathcal H$ given by $Sh=\sum_{j=1}^n\langle h,h_j\rangle h_j$.
\end{description}
\begin{Th}
A frame $H$ is tight if and only if its frame operator $S$ is proportional to identity, that is, $S=\alpha I$ for some $\alpha>0$. In particular, in such case \begin{equation}\label{traceformula}
    \alpha=\frac{1}{\dim \mathcal H}\sum_{j=1}^n\|h_j\|^2.
\end{equation}
\end{Th}

The equality (\ref{traceformula}) is derived by calculating the trace of the frame operator $S$ and thus sometimes referred to as the {\it trace formula}.

Since a frame is a spanning sequence, the associated synthesis operator $X^*$ (and thus the analysis operator $X$ as well) is always full-rank. Let $Y^*$ be any left inverse of $X$, that is, $Y^*X=I$ and let  $h_j':=Y^*\varepsilon^{(j)}$, where $\varepsilon^{(1)},\ldots,\varepsilon^{(n)}$ denotes the canonical basis of $\mathbb K^n$. Obviously, $H':=(h_j')_{j=1}^n$ is also a frame for $\mathcal H$ and $Y^*$ is its associated synthesis operator. Moreover, $X^*Y=(Y^*X)^*=I^*=I$. We can distinguish a special case, in which $Y^*$ is the pseudoinverse of $X$, that is, $Y^*=(X^*X)^{-1}X^*=S^{-1}X^*$, and, in consequence $h_j'=S^{-1}h_j$. These considerations allow us to put the following definitions:
\begin{Df} 
We say that $(h_j')_{j=1}^n$ is a {\bf dual frame} of $(h_j)_{j=1}^n$ if  
\begin{equation}
    h=\sum_{j=1}^n\langle h,h_j\rangle h_j'=\sum_{j=1}^n\langle h,h_j'\rangle h_j
\end{equation}
for all $h\in\mathcal H$. In particular, we say that $(h_j')_{j=1}^n$ is the {\bf canonical dual frame} if $h_j'=S^{-1}h_j$ for $j=1,\ldots,n$.
\end{Df}

It is easy to see that the importance of dual frames is that they give us a recipe how to reconstruct a vector $h\in\mathcal H$ from the sequence of its inner products with the frame vectors $(\langle h,h_j\rangle)_{j=1}^n$. However, in general, this recipe requires finding an inverse of a matrix. It would be therefore desirable to have a class of frames for which the computational complexity of the task can be strongly reduced. Obviously, the tight frames are distinguished by being, up to the scalar factor, their own (canonical) dual frames.  Thus, the frames that are in some sense `simple' deformations of tight frames arise as the natural candidates for the class in question:
\begin{Df}
A frame $(h_j)_{j=1}^n$ for $\mathcal H$ is called {\bf scalable} if there exist nonnegative constants $(s_j)_{j=1}^n$ such that $(s_jh_j)_{j=1}^n$ is a tight frame for $\mathcal H$.
\end{Df}
\begin{R} The scalability of a frame can be equivalently expressed as the existence of a dual frame of a specific form. Indeed, since
\begin{equation}
    \sum_{j=1}^n\langle h,s_jh_j\rangle s_jh_j=\sum_{j=1}^n\langle h,h_j\rangle s_j^2h_j,
\end{equation}
$(h_j)_{j=1}^n$ is scalable with scales $(s_j)_{j=1}^n$ to a tight frame with the frame bound $A$ if and only if $\left(\frac{s_j^2}{A}h_j\right)_{j=1}^n$ is its dual frame.
\end{R}
\begin{R}
If $X$ is the analysis operator for $(h_j)_{j=1}^n$ and  
 $D_s:=\textnormal{diag}(s_1,s_2,\ldots,s_n)$, the analysis operator for $(s_jh_j)_{j=1}^n$ is given by  $D_sX$.
\end{R}
\section{Special classes of measurements in GGPTs}\label{classes}

One of the desirable properties of a measurement is the possibility of recovering a measured state from the statistics of the measurement outcomes, i.e. these statistics need to uniquely determine the pre-measurement state. Such property is known as the informational completeness:
\begin{Df}
A measurement $\pi$ is called {\bf informationally complete} (or just {\bf IC} for short) if $\pi(x)=\pi(y)$ implies $x=y$ for all $x,y\in B$.
\end{Df}

The following simple characterisation gives us an insight into the structure of IC measurements:

\begin{Th}\cite[Lemma 2.1]{SinStu92}\cite[Thm 10]{SzySlo23}
The following conditions are equivalent:
\begin{enumerate}[i.]
    \item $\pi$ is informationally complete,
		\item\label{framev*} $\lin\{\pi_j:j=1,\ldots,n\}=V^*$,
		\item $\lin\{\mathcal{P}_0(\pi_j):j=1,\ldots,n\}=V_0^*$,
		\item\label{framev} $\lin\{T^{-1}(\pi_j):j=1,\ldots,n\}=V$,
		\item$\lin\{P_0(T^{-1}(\pi_j)):j=1,\ldots,n\}=V_0$.
\end{enumerate}
\end{Th}
Let us observe, that the conditions {\it\ref{framev*}.} and {\it\ref{framev}.} mean that $(\pi_j)_{j=1}^n$ and $(T^{-1}(\pi_j))_{j=1}^n=(v_j)_{j=1}^n$ are frames for $V^*$ and $V$, respectively. Moreover, since $\langle v_j,x\rangle=\pi_j(x)$, the measurement map $\pi$ can be interpreted as the analysis operator for $(v_j)_{j=1}^n$. It means that in order to recover the initial state $x$ it suffices to find a dual frame of $(v_j)_{j=1}^n$. It seems however, that it might be too much to ask for since we are only interested in the states, i.e. the vectors occupying (a part of) an affine subspace in $V$. Indeed, for $x\in B$ we have $\pi_j(x)-\pi_j(m)=\pi_j(x-m)=\langle v_j,x-m\rangle=\langle P_0(v_j),P_0(x)\rangle$ and $(\pi_j(m))_{j=1}^n$ can be treated as some known constants. Therefore we can reduce the dimension by one and look for the dual frames of $(P_0(v_j))_{j=1}^n$ in $V_0$. The special class of measurements with the property that $(P_0(v_j))_{j=1}^n$ is a tight frame for $V_0$ (and so is $(\mathcal P_0(\pi_j))_{j=1}^n$ for $V_0^*$) has been thoroughly investigated first  in the quantum \cite{SloSzy20} and then in the GGPT context  \cite{SzySlo23} under the name of morphophoric measurements. It seems quite natural to consider now a broader class:

\begin{Df} We say that $\pi$ is an {\bf $s$-tight IC} measurement if  $\left(\mathcal P_0\left(\pi_j\right)\right)_{j=1}^n$ is a scalable frame for $V_0^*$. 
\end{Df}
\begin{R}
If $\pi$ is an $s$-tight IC measurement with scalability constants $(s_j)_{j=1}^n$, the frame bound $\alpha$ for $(s_j\mathcal P_0(\pi_j))_{j=1}^n$ can be calculated using the trace formula (\ref{traceformula}):
$$\alpha=\frac{1}{\dim V_0}\sum_{j=1}^ns_j^2\left(\|\pi_j\|^2-\frac{1}{\mu}(\pi_j(m))^2\right).$$
\end{R}

 It might be useful to think of $s$-tight IC measurements in terms of the corresponding vectors $v_j=T^{-1}(\pi_j)$ or states $w_j=v_j/e(v_j)$. Since $T$ is not only an isometric isomorphism between $V$ and $V^*$ but also between $V_0$ and $V_0^*$, the proof of the  following statement can be easily derived:
 \begin{Prop} The following conditions are equivalent:
\begin{enumerate}[i.]
    \item $\pi$ is an $s$-tight IC measurement with scales $(s_j)_{j=1}^n$,
    \item $\left(P_0\left(v_j\right)\right)_{j=1}^n$ is a scalable frame for $V_0$ with scales $(s_j)_{j=1}^n$,
    \item $(w_j-m)_{j=1}^n$ is a scalable frame for $V_0$ with scales $(s_je(v_j))_{j=1}^n=(s_j\pi_j(m)/\mu))_{j=1}^n$.
\end{enumerate}
\end{Prop}

Obviously, the morphophoric measurements are $s$-tight IC with scales $s_j=const$. However, it turns out that there is another subclass of $s$-tight IC measurements that deserves a special attention: the ones with the scalability constants $s_j=1/\sqrt{\pi_j(m)}$, $j=1,\ldots,n$. Such measurements has been already considered (in the context of linear quantum tomography) by Scott \cite{Sco06} in quantum setup under the name of \emph{tight IC POVMs}. We adapt this nomenclature.

\begin{Df} We say that $\pi$ is a {\bf tight IC} measurement if it is $s$-tight IC with scales $\left(1/\sqrt{\pi_j(m)}\right)_{j=1}^n$. The frame bound $\alpha$ then simplifies to
\begin{equation}\label{tightICbound}
    \alpha=\frac{1}{\dim V_0}\left(\sum_{j=1}^n\frac{\|\pi_j\|^2}{\pi_j(m)}-\frac{1}{\mu}\right)=\frac{1}{\mu\dim V_0}\left(\sum_{j=1}^n\pi_j(w_j)-1\right).
\end{equation}
\end{Df}

\begin{R}\label{tightmorph}
If a measurement $\pi$ is \emph{unbiased}, i.e. $\pi_j(m)=1/n$ for $j=1,\ldots,n$, then obviously $\pi$ is tight IC if and only if it is morphophoric. However, these are not the only measurements that are both morphophoric and tight IC. Let us recall \cite[Prop. 15.i.]{SzySlo23} that if $\pi^1,\ldots,\pi^m$ are morphophoric measurements, then also $\pi:=(t_1\pi^1)\cup\ldots\cup(t_m\pi^m)$ is morphophoric for any $t_1,\ldots,t_m\geq 0$ such that $t_1+\ldots+t_m=1$. It is easy to see that if $\pi^1,\ldots,\pi^m$ are additionally unbiased (and so tight IC), the measurement $\pi$ is always tight IC, although it does not need to be unbiased. The relations between morphophoric, tight IC and $s$-tight IC measurements are visualised in the Figure \ref{figclasses}.\end{R}

\begin{figure}[htb]
    \centering
    \includegraphics[scale=0.4]{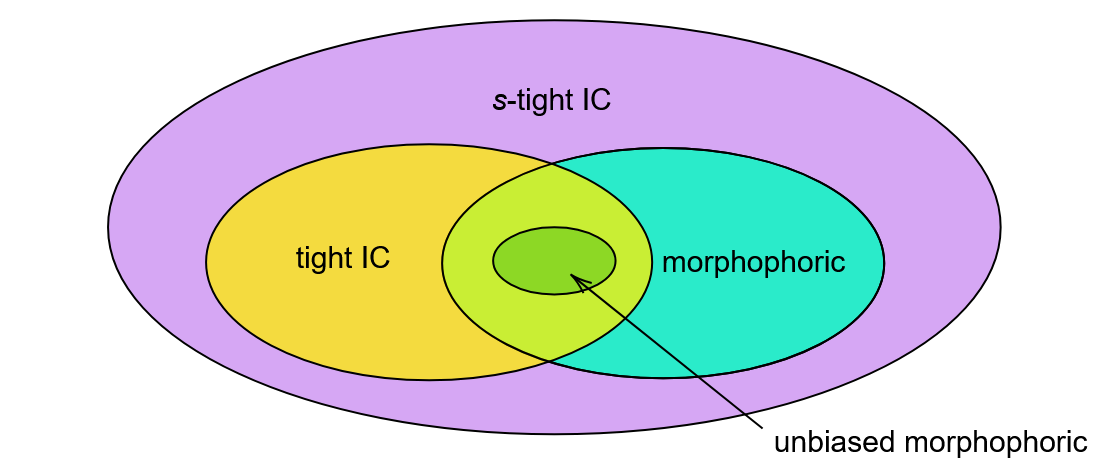}
    \caption{Relations between different classes of measurements.}
    \label{figclasses}
\end{figure}

It turns out that tight IC measurements have an additional interesting property which can be obtained as a generalisation of \cite[eq. (50)]{Sco06}.  

\begin{Prop}\label{tightproperty}
For a  measurement $\pi$ let $S$ denote the frame operator for $\left(\frac{1}{\sqrt{\pi_j(m)}}v_j\right)_{j=1}^n$ and $S_0$ -- the frame operator for $\left(\frac{1}{\sqrt{\pi_j(m)}}P_0(v_j)\right)_{j=1}^n$. Then $S=S_0+\frac{1}{\mu}P_m$, where $P_m=I-P_0$ is the orthogonal projection onto $\textnormal{lin}\{m\}$. Moreover, $\pi$ is a tight IC measurement if and only if $S=\alpha P_0+\frac{1}{\mu}P_m$ for some $\alpha>0$. 
\end{Prop}
\begin{proof}
Let us first observe that $SP_m=\frac{1}{\mu}P_m$. Indeed, for $x\in V$ we have $$SP_m(x)=e(x)S(m)=e(x)\sum_{j=1}^n\frac{1}{\pi_j(m)}\langle m,v_j\rangle v_j=e(x)\sum_{j=1}^n v_j=\frac{e(x)}{\mu}m=\frac{1}{\mu}P_m(x).$$
In consequence,
$$S=P_0SP_0+P_0SP_m+P_mSP_0+P_mSP_m=S_0+\frac{1}{\mu}P_0P_m+\frac{1}{\mu}P_mP_0+\frac{1}{\mu}P_m^2=S_0+\frac{1}{\mu}P_m.$$
The `moreover' part is straightforward from the definition of a tight IC measurement.
\end{proof}

From the frame theoretical point of view this property can be equivalently expressed by saying that $\left((1/\sqrt{\pi_j(m)})v_j\right)_{j=1}^n$ is a \emph{lift} of $\left((1/\sqrt{\pi_j(m)})
P_0(v_j)\right)_{j=1}^n$, see \cite[Df. 5.2]{Wal18}.

In quantum theory, one often distinguishes measurements that are \emph{rank 1}, i.e. the effects  constituting such measurement are rank 1 operators. In particular, it means that $\arg\max_{x\in B}\pi_j(x)$ is a singleton for $j=1,\ldots,n$. In the language of self-dual generalised probabilistic theories it also means that these effects lie on the extreme rays of the dual cone (note that these two properties do not coincide in general for GPTs that are not self-dual). But since quantum theory is equinorm, all the extreme rays are also maximally distant from the central ray $\{tI:t\geq 0\}$. Thus we consider the following to be a generalisation of rank 1 measurements in the general (not necessarily equinorm) self-dual setup:
\begin{Df}
We say that a measurement $\pi$ in a self-dual GGPT is a {\bf $\chi$-ray} measurement if the effects $\pi_j$, $j=1,\ldots,n$, lie on the rays of the dual cone maximally distant from the central ray $\{te:t\geq 0\}$.
\end{Df}
\begin{R}\label{rank1R}
The reason why we call these measurements $\chi$-ray is that they satisfy  $\|w_j-m\|^2_0=\chi$, $j=1,\ldots,n$.
\end{R} 
It turns out that in such case we can bind together three constants of the theory: the \emph{measurement constant} $\alpha\mu$, the \emph{space constant} $\chi/\mu$ and the \emph{space dimension} $\dim V_0$ in a similar way that it was done for the morphophoric regular measurements in \cite[Theorem 23]{SzySlo23}. However, unlike in that  case, this time the \emph{measurement dimension} $n$ is not included.
\begin{Th}\label{rank1constants}
Let $\pi$ be a $\chi$-ray tight IC measurement in a self-dual GGPT. Then
\begin{equation*}
    \alpha\mu=\frac{\chi/\mu}{\dim V_0}.
    \end{equation*}
\end{Th}
\begin{proof}
By applying (\ref{tightICbound}) and Remark \ref{rank1R} we get
\begin{align*}\alpha\mu&=\frac{1}{\dim V_0}\left(\sum_{j=1}^n\pi_j(w_j)-1\right)=\frac{1}{\dim V_0}\left(\sum_{j=1}^ne(v_j)\|w_j\|^2-1\right)=\frac{1}{\dim V_0}\left(\sum_{j=1}^ne(v_j)(\chi+\mu)-1\right)\\&=\frac{1}{\dim V_0}\left(\frac{\chi+\mu}{\mu}-1\right)=\frac{\chi}{\mu\dim V_0}.
\end{align*}
\end{proof}

\begin{Ex}
Let us consider a 3-parameter family of informationally complete POVMs on $\mathbb C^2$ given by 
\begin{equation*}
    \Pi_1=aI+bZ,\ \Pi_2=aI-bZ,\ \Pi_3=\frac{1-2a}{3}I+cX,\ \Pi_4=\frac{1-2a}{3}I-\frac{c}{2}X+\frac{\sqrt{3}c}{2}Y,\ \Pi_5=\frac{1-2a}{3}I-\frac{c}{2}X-\frac{\sqrt{3}c}{2}Y,
\end{equation*} where $a\in (0,\frac{1}{2})$, $b\in (0,a]$, $c\in (0,\frac{1-2a}{3}]$ (these are the maximal ranges of parameters, for which the operators $\Pi_j$ are positive semidefinite and the resulting POVM is informationally complete), and $X,Y$ and $Z$ denote the Pauli matrices. In particular, such POVM is rank 1 for $b=a$ and $c=\frac{1-2a}{3}$. The projections $P_0(\Pi_j)$ onto the 3-dimensional traceless subspace can be indentified as the following vectors:
\begin{equation*}
    (0,0,b),\ (0,0,-b),\ (c,0,0),\ (-\frac{c}{2},\frac{\sqrt{3}c}{2},0),\ (-\frac{c}{2},-\frac{\sqrt{3}c}{2},0).
\end{equation*}
It is easy to see that these vectors constitute a scalable frame since they can be seen as a disjoint union of tight frames for $\mathbb R$ and $\mathbb R^2$. Such POVM is thus $s$-tight IC for any choice of the parameters $a,b$ and $c$, and also:
\begin{itemize}
    \item morphophoric if $c=\frac{2\sqrt{3}b}{3}$
    \item rank 1 morphophoric if $a=b=\frac{\sqrt{3}-1}{4}$ and $c=\frac{3-\sqrt{3}}{6}$
    \item tight IC if $c^2=\frac{4b^2-8ab^2}{9a}$
    \item rank 1 tight IC if $a=b=\frac{1}{6}$ and $c=\frac{2}{9}$
    \item morphophoric tight IC if $a=\frac{1}{5}$ and $c=\frac{2\sqrt{3}b}{3}$.
\end{itemize}
The structure of the set of parameters is presented in Figure \ref{parameters}.
\begin{figure}[htb]
   \centering
    \includegraphics[scale=0.5]{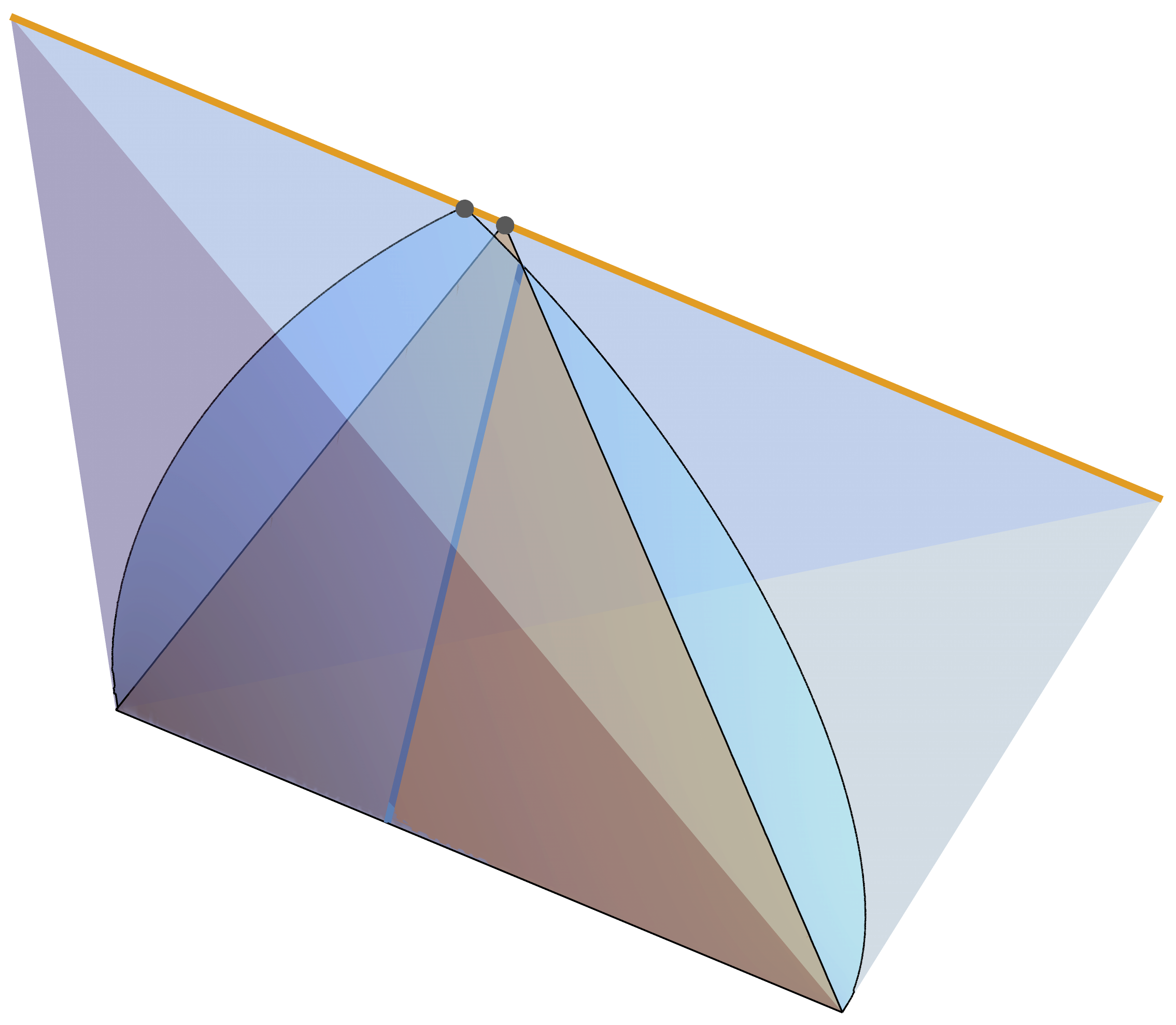}
    \caption{The polyhedral set of parameters for which the POVM is $s-tight$. The intersection with a plane (orange triangle) indicates the set of parameters for which the  POVM is morphophoric while the intersection with a blue surface indicates the set of parameters for which it is tight IC. The blue segment is the intersection of these two surfaces and depicts the cases in which the POVM in question is both morphophoric and tight IC. The orange segment at the edge of the polyhedron corresponds to rank 1 POVMs. Two distinguished points on this segment correspond to rank 1 tight IC and rank 1 morphophoric POVMs, respectively. It is clear that in this family of POVMs there is no possibility to obtain a POVM that is rank 1, tight IC and morphophoric at the same time.}
    \label{parameters}
\end{figure}

\end{Ex}
\section{Generalised Urgleichung}\label{genurg}
Our motivation in introducing $s$-tight IC measurements was to provide a tool that would give us a simple recipe for recovering a state from the statistics of the measurement outcomes. And if we can easily recover the state, we can also easily recover the statistics of \emph{any} measurement without actually performing it. In consequence we obtain one more generalisation of the \emph{primal equation} or \emph{Urgleihung}.
The primal equation in QBism does not only play the role of the state (or statistics) recovering formula. In general it also shows the nonclassicality of the measurement (or theory) by the visible difference from the classical Law of Total Probability (LTP). If we want to interpret this recovering formula as some  generalisation of the LTP, we need to know what happens to the system after measurement so that we can calculate the conditional probabilities. Thus we need to define a measurement instrument:

\begin{Df}
Let $\pi:= (\pi_j)^n_{j=1}$ be a measurement and let $\Lambda:=(\Lambda_j)_{j=1}^n$ be a family of affine maps from
$B$ to $C$. We say that $\Lambda$ is an {\bf instrument} for $\pi$ if $\pi_j (x) = e(\Lambda_j (x))$ for every pre-measurement state
$x\in B$ and for all $j = 1,\ldots,n$. We assume that the {\bf post-measurement state} is given by $\Lambda_j (x)/e(\Lambda_j (x))$,
supposing that the result of measurement was actually $j$ (and so $\pi_j (x)\neq 0$). We say that the instrument is {\bf balanced at $m$} if $\Lambda_j(m)=\mu v_j=\pi_j(m)w_j$. In particular, if $\Lambda_j(x)=\pi_j(x)w_j$ for every $x\in B$, we call such instrument {\bf canonical}.
\end{Df}

If we perform two measurements one after another, we can express the conditional probabilities of their results by using the measurement instrument for the first one. Namely, if we denote by $p^{\xi|\pi}_{k|j}(x)$ the conditional probability that the result of the measurement $\xi=(\xi_k)_{k=1}^{n'}$ is $k$ given that we measured first $j$ with the measurement $\pi=(\pi_j)_{j=1}^n$ performed on the initial state $x$, then $p^{\xi|\pi}_{k|j}(x)=\xi_k(\Lambda_j(x)/e(\Lambda_j(x))$ for $j=1,\ldots,n'$ and $k=1,\ldots,n$.

The original QBist approach is to consider the conditional probabilities coming from a \emph{conditional state preparator}, i.e. an instrument of the form $\Lambda_j(x)=\pi_j(x)y_j$ for some arbitrary ensemble $(y_j)_{j=1}^n$ of states (such measurement -- measurement instrument setup is called a \emph{reference device}). However, it is worth mentioning that despite initially considering this general form of conditional state preparator, in the QBism-related papers the authors usually quickly restrict their attention to a special case: the \emph{canonical}  instrument, which used to be justified by the `aesthetic reasons' and then by some optimality criterion \cite[Section 2.5]{Fuc23}. But since the canonical instruments are also special cases of another, much more general family: instruments \emph{balanced at} $m$,  it makes sense to consider for the generalised primal equations the conditional probabilities coming from an instrument that is balanced at $m$ rather than from the general conditional state preparator. When we put it this way, it turns out that the only conditional probabilities needed to derive the generalised primal equation are the ones for the initial state $m$. What is more,  the `balanced at $m$' condition guarantees that the probabilities for the initial state $m$ follow the classical law of total probability \cite[Eq. (26)]{SzySlo23}:
\begin{equation}\label{balancedatm}
    p_k^\xi(m)=\sum_{j=1}^n p_j^\pi(m) p_{k|j}^{\xi|\pi}(m).
\end{equation} All this reveals a consequence of one of the main differences between these two approaches: with general conditional state preparator and with an instrument balanced at $m$. In the latter case the post-measurement states do not depend on the initial state of the system, creating an illusion of imitating the law of total probability\footnote{The more detailed discussion can be found in \cite[Section 6.2]{SzySlo23}}.   

We start with deriving the generalised primal equation first for $s$-tight IC  and then for tight IC measurements under the assumption that the measurement instrument is balanced at $m$. We begin with the operator form of the equations and then translate them in the form of corollaries into probabilistic language, that should be easier to interpret.

\begin{Th}[Generalised primal equation -- $s$-tight IC case]\label{urgthm} Let $\pi=(\pi_j)_{j=1}^n$ be a measurement in a supra-dual GGPT with an
instrument $\Lambda$ balanced at $m$. Then the following conditions are equivalent:
\begin{enumerate}[i.]
    \item $\pi$ is an $s$-tight IC measurement with scales $(s_j)_{j=1}^n$
    \item for every measurement $\xi=(\xi_k)_{k=1}^{n'}$ \begin{equation}\label{urgscal}
        \delta_\xi=\frac{1}{\alpha\mu}\mathsf C\delta_{\pi}
    \end{equation}
    holds, where $\delta_\xi=\xi\circ P_0$, $\delta_\pi=\pi\circ P_0$,  $\mathsf C_{kj}=s_j^2(\xi_k(\Lambda_j(m))-\pi_j(m)\xi_k(m))$, and  $$\alpha=\frac{1}{\dim V_0}\sum_{j=1}^ns_j^2\left(\|\pi_j\|^2-\frac{1}{\mu}(\pi_j(m))^2\right)$$
    \item for some informationally complete measurement  $\xi=(\xi_k)_{k=1}^{n'}$ and some $\alpha>0$ $$\delta_\xi=\frac{1}{\alpha\mu}\mathsf C\delta_{\pi}$$
    holds, where  $\delta_\xi$, $\delta_\pi$ and  $\mathsf C_{kj}$  are as above.
\end{enumerate}
\end{Th}

\begin{proof}
Let us first observe that $\delta_\pi=\pi\circ P_0$ is the analysis operator for $(P_0(v_j))_{j=1}^n$. Indeed, for $x\in V$ we have $(\delta_\pi(x))_j=\pi_j(P_0(x))=\langle v_j,P_0(x)\rangle=\langle P_0(v_j),x\rangle$. Thus $D_s\delta_\pi$ is the analysis operator for $(s_jP_0(v_j))_{j=1}^n$. Obviously, $\xi$ is the analysis operator for $(T^{-1}(\xi_k))_{k=1}^{n'}$. Now,
\begin{align*}
    (\xi\delta_\pi^*)_{kj}&=\langle T^{-1}(\xi_k),P_0(v_j)\rangle=\xi_k(P_0(v_j))=\xi_k(v_j)-e(v_j)\xi_k(m)=\frac{1}{\mu}\xi_k(\Lambda_j(m))-\frac{1}{\mu}\pi_j(m)\xi_k(m)
\end{align*}
In consequence, $\mathsf{C}=\mu(\xi\delta_\pi^*)D_s^2$ and formula (\ref{urgscal}) can be written equivalently as \begin{equation}\label{urgscal2}\xi P_0=\frac{1}{\alpha}\xi\delta_\pi^*D_s^2\delta_\pi=\frac{1}{\alpha}\xi(D_s\delta_\pi)^*(D_s\delta_\pi).\end{equation}
Now, condition i. can be equivalently expressed as $(D_s\delta_\pi)^*(D_s\delta_\pi)=\alpha P_0$. Thus, by applying $\xi$ on both sides we obtain the implication i. $\implies$ ii. The implication ii. $\implies$ iii. is obvious and finally the implication iii.$\implies$ i. is obtained by applying the pseudoinverse $\xi^+$ on both sides of equation \ref{urgscal2}.
\end{proof}

\begin{Cor}[Probabilistic version] Let $\pi=(\pi_j)_{j=1}^n$ be an $s$-tight IC  measurement with scales $(s_j)_{j=1}^n$ in a supra-dual GGPT with an instrument $\Lambda$ balanced at $m$, and let $\xi=(\xi_k)_{k=1}^{n'}$ be an arbitrary measurement. The equation (\ref{urgscal}) can be written in probabilistic terms (we put $p_j^\pi:=\pi_j$ and $p_k^\xi:=\xi_k$) as
$$p_k^\xi(x)-p_k^\xi(m)=\dim V_0\cdot \frac{\sum_{j=1}^ns_j^2p_j^\pi(m)\left(p_{k|j}^{\xi|\pi}(m)-p_k^\xi(m)\right)\left(p_j^\pi(x)-p_j^\pi(m)\right)}{\sum_{j=1}^ns_j^2p_j^\pi(m)\left(p_{j|j}^{\pi|\pi}(m)-p_j^\pi(m)\right)}	$$
for $k=1,\ldots,k'$ and $x\in B$.
\end{Cor}

One can argue that unlike in the morphophoric case, this `probabilistic' version  is not purely probabilistic in the sense that it involves some arbitrary factors $s_j$ that cannot be interpreted in the probabilistic way. Thus a natural question arises: can we obtain a nice, purely probabilistic formula for some class of $s$-tight IC measurements other than morphophoric ones? The answer is yes and it is easy to guess that one should consider $s$-tight IC measurements with scales $s_j=1/\sqrt{p_j^\pi(m)}$, i.e. the tight IC measurements.

\begin{Th}[Generalised primal equation -- tight IC case]\label{urgtight} Let $\pi=(\pi_j)_{j=1}^n$ be a measurement in a supra-dual GGPT with an
instrument $\Lambda$ balanced at $m$. TFAE:
\begin{enumerate}[i.]
    \item $\pi$ is a tight IC measurement
    \item for every measurement $\xi=(\xi_k)_{k=1}^{n'}$ \begin{equation}\label{urgtighteq}
        \delta_\xi=\frac{1}{\alpha\mu}\mathsf K\delta_{\pi}
    \end{equation}
    holds, where $\delta_\xi=\xi\circ P_0$, $\delta_\pi=\pi\circ P_0$,  $\mathsf K_{kj}=\xi_k(\Lambda_j(m)/e(\Lambda_j(m))=\xi_k(w_j)$, and $$\alpha=\frac{1}{\mu\dim V_0}\left(\sum_{j=1}^n\pi_j(w_j)-1\right)$$
    \item for some informationally complete measurement $\xi=(\xi_k)_{k=1}^{n'}$ and some $\alpha>0$ $$\delta_\xi=\frac{1}{\alpha\mu}\mathsf K\delta_{\pi}$$
    holds, where  $\delta_\xi$, $\delta_\pi$ and  $\mathsf K_{kj}$ are as above.
\end{enumerate}
\end{Th}
\begin{proof}
It suffices to show that for $s_j=1/\sqrt{\pi_j(m)}$, $j=1,\ldots,n$ we can express $\mathsf C\delta_\pi$ from Theorem \ref{urgthm} as $\mathsf K\delta_\pi$. By the observations made in the proof of Theorem \ref{urgthm} and by Proposition \ref{tightproperty} we obtain
$$\mathsf C\delta_\pi=\mu\xi\delta_\pi^*D_s^2\delta_\pi=\mu\xi(\pi^*D_s^2\pi-\frac{1}{\mu}P_m)P_0=\mu\xi\pi^*D_s^2\delta_\pi.$$
But $(\xi\pi^*)_{kj}=\xi_k(v_j)=\frac{1}{\mu}\pi_j(m)\xi_k(w_j)$, and thus $\mu\xi\pi^*D_s^2=\mathsf K$ and, in consequence, $\mathsf C\delta_\pi=\mathsf K\delta_\pi$.
\end{proof}

\begin{Cor}[probabilistic version] Let $\pi=(\pi_j)_{j=1}^n$ be a tight IC measurement  in a supra-dual GGPT with an instrument $\Lambda$ balanced at $m$, and let $\xi=(\xi_k)_{k=1}^{n'}$ be an arbitrary measurement. The equation (\ref{urgtighteq}) can be written in probabilistic terms (we put $p_j^\pi:=\pi_j$ and $p_k^\xi:=\xi_k$) as
$$p_k^\xi(x)-p_k^\xi(m)=\dim V_0\cdot\frac{\sum_{j=1}^np_{k|j}^{\xi|\pi}(m)(p_j^\pi(x)-p_j^\pi(m))}{\sum_{j=1}^n(p_{j|j}^{\pi|\pi}(m)-p_j(m))}$$
\end{Cor}
It is worth to notice that the Urgleichung for tight IC measurements (both the general form and the purely probabilistic one) is much simpler than one could expect by simply replacing $s_j^2$ with $1/\pi_j(m)$.

We have shown previously \cite[Proposition 29]{SzySlo23}  that for unbiased morphophoric measurements and canonical instruments the Urgleichung takes the form that can be translated into the one reminding the total probability law \cite[Remark 13]{SzySlo23} (and indeed reduces to it for classical probability theory). It turns out however, that this particular form is not related to the unbiasedness of the measurement but rather to its tightness (note that unbiased morphophoric measurements are also tight IC, see Remark \ref{tightmorph}):
\begin{Th}[Generalised primal equation -- canonical instrument]
Let $\Lambda$ be the canonical instrument for a measurement $\pi$ in a supra-dual GGPT. Then the following conditions are equivalent:
\begin{enumerate}[i.]
    \item $\pi$ is a tight IC measurement
    \item for every measurement $\xi=(\xi_k)_{k=1}^{n'}$ \begin{equation}\label{LTP}
        p_k^\xi(x)-p_k^\xi(m)=A\sum_{j=1}^np_{k|j}^{\xi|\pi}(x)(p_j^\pi(x)-p_j^\pi(m))
    \end{equation}
    holds for $x\in B$ and $k=1,\ldots,n'$ with $A=\frac{1}{\alpha\mu}$
    \item for some informationally complete measurement $\xi=(\xi_k)_{k=1}^{n'}$ $$p_k^\xi(x)-p_k^\xi(m)=A\sum_{j=1}^np_{k|j}^{\xi|\pi}(x)(p_j^\pi(x)-p_j^\pi(m))$$
    holds for $x\in B$ and $k=1,\ldots,n'$ with $A=\frac{1}{\alpha\mu}$
\end{enumerate}
\end{Th}

\begin{proof}
It follows from the probabilistic version of Theorem \ref{urgtight} and the fact that for the canonical instrument $p_{k|j}^{\xi|\pi}(x)=const(j)$.
\end{proof}

\begin{R}
Alternatively, (\ref{LTP}) can be rewritten in the following equivalent forms:
\begin{equation}
    p_k^\xi(x)=\sum_{j=1}^np_{k|j}^{\xi|\pi}(x)p_j^\pi(x)+(1-1/A)(p_k^\xi(x)-p_k^\xi(m))
    \end{equation}
or
\begin{equation}
    p_k^\xi(x)=\sum_{j=1}^np_{k|j}^{\xi|\pi}(x)(Ap_j^\pi(x)+(1-A)p_j^\pi(m))
\end{equation}
for $x\in B$ and $k=1,\ldots,n'$ with $A=\frac{1}{\alpha\mu}$ The former may be seen as the \emph{classical total probability law plus a correction term}. The latter can be derived thanks to Equation (\ref{balancedatm}) and it strikingly reminds the original QBist Generalised Urgleichung stated in the \cite[Assumption 1]{FucSch11}. The dependency on the initial state $x$ is omitted there, but in our notation it would take the following form: $p_k^\xi(x)=\sum_{j=1}^np_{k|j}^{\xi|\pi}(x)(Ap_j^\pi(x)+\beta)$. It is now easy to see that the only difference is that instead of a constant $\beta$ we have $\beta_j:=(1-A)p_j^\pi(m)$, depending only on the outcome $j$ but not on the initial state $x$.   \end{R}

\begin{R}
If the GGPT is self-dual and the tight IC measurement is additionally $\chi$-ray then it follows from Theorem \ref{rank1constants} that $A=\frac{\mu\dim V_0}{\chi}$. Let us note that this constant is exactly the same as the one for morphophoric regular measurements in self-dual GGPTs. It is obviously not surprising since the latter ones form a special subclass of $\chi$-ray tight IC, but it emphasizes once again that in the Urgleichung context, the tight IC measurements should be the ones to look for. 
\end{R}

\section{Conclusions}

A requirement that the state of the system can be fully described by the measurement statistics is easily satisfied with the informational completeness. However, in the vast landscape of IC measurements, we have managed to identify these for which the task of the state recovery is quite simple and, in consequence, for which we can deliver a generalised version of the primal equation. Most importantly, we have identified tight IC measurements as the only ones for which the generalised primal equation takes the form of the classical total probability law with a correction term, given that the measurement instrument is canonical.

\section*{Acknowledgments}

I would like to thank Wojciech S\l omczy\'nski for valuable discussions and comments.

\end{document}